\documentclass[conference]{IEEEtran}
\IEEEoverridecommandlockouts

\usepackage{graphicx,cite,amssymb,amsmath,psfrag,subfigure}
\usepackage{graphicx}
\usepackage{amssymb}
\usepackage{amsmath}
\usepackage{cite}
\usepackage{subfigure}
\usepackage{mathrsfs}
\usepackage[displaymath,mathlines]{lineno}
\usepackage{color}
\usepackage{tabulary}
\usepackage{multirow}
\usepackage{algpseudocode}
\usepackage{algorithm,algpseudocode}
\usepackage{algorithmicx}
\usepackage{pbox}
\usepackage{multicol}
\usepackage{lipsum}

\newtheorem{theorem}{Theorem}

\newtheorem{corollary}{Corollary}

\ifCLASSINFOpdf
  
\else
 
\fi

\begin{document}

\title{Downlink Power Control for Massive MIMO Cellular Systems with Optimal User Association}

\author{\IEEEauthorblockN{ Trinh Van Chien, Emil Bj\"{o}rnson, and Erik G. Larsson}
\IEEEauthorblockA{Department of Electrical
    Engineering (ISY), Link\"{o}ping University, 581 83 Link\"{o}ping, Sweden\\
\{trinh.van.chien, emil.bjornson, erik.g.larsson\}@liu.se}
\thanks{This paper was supported by the European Union's Horizon 2020 research and innovation programme under grant agreement No 641985 (5Gwireless). It was also supported by ELLIIT and CENIIT.}
}

\maketitle

\begin{abstract}
This paper aims to minimize the total transmit power consumption for Massive MIMO (multiple-input multiple-output) downlink cellular systems when each user is served by the optimized subset of the base stations (BSs). We derive a lower bound on the ergodic spectral efficiency (SE) for Rayleigh fading channels and maximum ratio transmission (MRT) when the BSs cooperate using non-coherent joint transmission. We solve the joint user association and downlink transmit power minimization problem optimally under fixed SE constraints. Furthermore, we solve a max-min fairness problem with user specific weights that maximizes the worst SE among the users. The optimal BS-user association rule is derived, which is different from maximum signal-to-noise-ratio (max-SNR) association. Simulation results manifest that the proposed methods can provide good SE for the users using less transmit power than in small-scale systems and that the optimal user association can effectively balance the load between BSs when needed.
\end{abstract}

\IEEEpeerreviewmaketitle

\section{Introduction}
By $2017$ there will be seven trillion wireless devices and the fast growth will increase the global $\textrm{CO}_2$-equivalent emissions significantly \cite{Wang2014,Auer2011a}. Moreover, $80 \%$ of the power in current networks is consumed at the BSs \cite{Auer2011a}. The BS technology therefore needs to be redesigned to reduce the power consumption.
Many researchers have investigated how the physical layer transmissions can be optimized to reduce the transmitted signal power, while maintaining the quality-of-service (QoS); see \cite{   Sun2015, Bjornson2013d, Li2015} and references therein. However, the papers \cite{ Sun2015, Bjornson2013d, Li2015} are all optimizing the power with respect to the small-scale fading, which is practically questionable since the fading coefficients change rapidly (i.e., the powers must be reoptimized every few milliseconds) and since the fading can be efficiently mitigated by channel coding. In contrast, the small-scale fading has negligible impact on Massive MIMO systems, thanks to favorable propagation \cite{Bjornson2016b}, and closed-form expressions for the ergodic SE are available for linear precoding schemes \cite{Ngo2013a}. The power allocation can be optimized with respect to the large-scale fading in Massive MIMO \cite{Larsson2014a}, which makes advanced power control algorithms computationally feasible. A few recent works have considered power allocation for Massive MIMO systems \cite{Guo2014a,Victor2015b}, but none of them has considered the BS-user association problem.

Massive MIMO has demonstrated high energy efficiency in homogeneously loaded scenarios \cite{Ngo2013a}, where an equal number of users are preassigned to each BS. At any given time, the user load is typically heterogeneously distributed in practice, such that some BSs have many more users in their vicinity than others. Large SE gains are possible by balancing the load over the network \cite{Ye2013a}, using other user association rules than simple max-SNR association. The optimal association is naturally a combinatorial problem with a complexity that scales exponentially with the network size \cite{Ye2013a}. While load balancing is a well-studied problem for heterogeneous multi-tier networks, the recent work \cite{Bjornson2013e} has shown that large gains are possible also in Massive MIMO systems with heterogeneous user loads.

In this paper, we jointly optimize the downlink (DL) power allocation and BS-user association for Massive MIMO cellular systems. A key assumption is that each user can be served by multiple BSs, using low-complexity non-coherent joint transmission. First, we derive a new general ergodic SE expression for the scenario when the signals are decoded in a successive manner. A closed-form expression is then derived for MRT precoding and Rayleigh fading channels. After that, we formulate a long-term power minimization problem under constraints on the ergodic SE per user and maximum transmit power per BS. This is shown to be a linear program when the new ergodic SE expression for MRT precoding is used, so the optimal solution can be obtained in polynomial time. The solution also provides the optimal BS-user association policy, which assigns a single BS per user in most cases. In addition to fixing the target SE constraints, we consider weighted max-min SE optimization and show that this problem can also be solved efficiently using our new SE expression.

\textit{Notations:}  We use the upper-case bold face letters for matrices and lower-case bold face ones for vectors. $\mathbf{I}_M$ and $\mathbf{I}_K$ are the identity matrices of size $M \times M$ and $K \times K$, respectively. The operator $\mathbb{E} \left\{ \cdot \right\}$ represents  the expectation of a random variable. The notation $ \| \cdot \| $ stands for the Euclidean norm. The regular and Hermitian transposes are denoted by $(\cdot)^T$ and  $(\cdot)^H$, respectively. Finally, $\mathcal{CN}(\cdot,\cdot)$ represents the circularly symmetric complex Gaussian distribution.
\begin{figure}[t]
    \centering
    \includegraphics[trim=3.5cm 7.0cm 3cm 10cm, clip=true, width=2.6in]{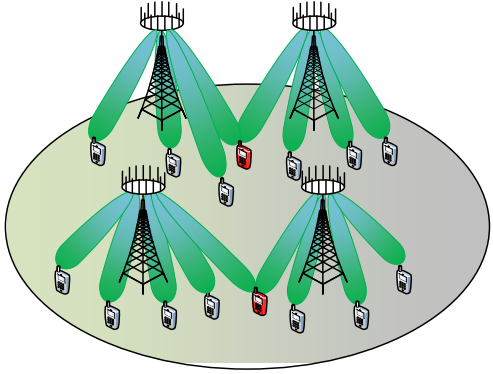}
    \caption{A multiple-cell Massive MIMO DL system where users can be associated with more than one BS (e.g., red users). The optimized BS subset for each user is obtained from the proposed optimization problem.} 
    \label{fig:MassiveMIMOSystems}
    \vspace{-0.4cm}
\end{figure}

\section{System Model and Achievable Performance} \label{Section:System-Model and Achievable Performance}

A schematic diagram of our system model is shown in Fig.~\ref{fig:MassiveMIMOSystems}. We consider a Massive MIMO system with $L$ cells. Each cell comprises a BS with $M$ antennas. The system serves $K$ single antenna users in the same time-frequency resource. Note that each user is conventionally associated and served by only one of the BSs. However, in this paper, we optimize the BS-user association, so the users are numbered from $1$ to $K$ without having predefined cell indices. Each user is preassociated with all BSs, while the optimal subset of BSs is later found by optimization. We assume that the channels are constant and frequency-flat in a coherence interval of length $\tau_c$ symbols and the system operates in time division duplex (TDD) mode. In each coherence interval, $\tau_p$ symbols are used for channel estimation, and $\tau_c - \tau_p$ symbols are dedicated for the data transmission. We focus on the DL data transmission in this paper.
\vspace*{-0.2cm}
\subsection{Uplink Channel Estimation}
\vspace*{-0.2cm}
Assume that all users simultaneously transmit mutually orthogonal pilot sequences of length $\tau_p$ with $\tau_p \geq K$. The received pilot signal $\mathbf{Y}_l \in \mathbb{C}^{M \times \tau_p}$ at BS $l$ is expressed as

\begin{equation} \label{eq: UL-Recieved-Pilot-Matrix}
\mathbf{Y}_l = \mathbf{H}_l \mathbf{P}^{1/2} \pmb{ \Phi }^H + \mathbf{N}_l,
\end{equation}
where the channel matrix  $\mathbf{H}_l = [\mathbf{h}_{l,1}, \ldots,  \mathbf{h}_{l,K} ] \in \mathbb{C}^{M \times K}$ has the column vectors $\mathbf{h}_{l,k}$  each of which denotes the channel between user $k$ and BS $l$, for $k=1,\ldots,K$ and $l=1,\ldots, L$.  In this paper, we consider uncorrelated Rayleigh fading channels $\mathbf{h}_{l,k} \sim \mathcal{CN}(\mathbf{0}, \beta_{l,k} \mathbf{I}_M)$ with the channel variance $\beta_{l,k}$. The orthogonality of the pilot sequences implies that the $\tau_p \times K$ pilot matrix $\pmb{\Phi} =  [\pmb{\phi}_1, \ldots , \pmb{\phi}_K]$ satisfies $\pmb{\Phi}^H \pmb{\Phi} = \tau_p \mathbf{I}_K$. If we let $p_k$ denote the power that user $k$ assigns for each uplink (UL) pilot symbol, then the diagonal power matrix is formulated as $\mathbf{P} = \mathrm{diag} (p_1, \ldots, p_K) \in \mathbb{C}^{K \times K}$. Finally,  $\mathbf{N}_l \in \mathbb{C}^{M \times \tau_p}$ is Gaussian noise with its independent entries having the distribution $\mathcal{CN} (0, \sigma_{\mathrm{UL}}^2 )$. 

Based on the received pilot signal in \eqref{eq: UL-Recieved-Pilot-Matrix} and assuming that the BS knows the channel statistics, it can apply minimum mean square error (MMSE) estimation \cite{Kay1993a} to obtain an estimate $\hat{ \mathbf{h} }_{l,k}$ of $\mathbf{h}_{l,k}$. Due to the orthogonality of the pilot sequences, the channel $\mathbf{h}_{l,k}$ only occurs as $\mathbf{h}_{l,k} \pmb{\phi}_k^H $ in \eqref{eq: UL-Recieved-Pilot-Matrix}. Thus, a sufficient statistic to estimate channel $\mathbf{h}_{l,k}$ is
\begin{equation}
\begin{split}
\mathbf{Y}_l \pmb{\phi}_k  = \mathbf{H}_l \mathbf{P}^{1/2} \pmb{ \Phi }^H \pmb{\phi}_k  + \mathbf{N}_l \pmb{\phi}_k = \sqrt{p_{k}} \tau_p \mathbf{h}_{l,k} + \tilde{ \mathbf{n} }_{l,k},
\end{split}
\end{equation}
where $\tilde{ \mathbf{n} }_{l,k} = \mathbf{N}_l \pmb{\phi}_k \sim \mathcal{CN}( \mathbf{0}, \tau_p \sigma_{\mathrm{UL}}^2 \mathbf{I}_M )$. The MMSE estimate $\hat{ \mathbf{h} }_{l,k}$  of $\mathbf{h}_{l,k}$ has the estimation error defined as $\mathbf{e}_{l,k} = \hat{\mathbf{h}}_{l,k} - \mathbf{h}_{l,k} $. Consequently, the channel estimate and the estimation error are independent and distributed as
\begin{align}
\label{eq: Estimate-Channel-Distribution}
\hat{ \mathbf{h} }_{l,k} & \sim \mathcal{CN} \left(  \mathbf{0}, \frac{ p_k  \tau_p \beta_{l,k}^2 }{ p_k  \tau_p \beta_{l,k} + \sigma_{ \mathrm{UL} }^2  }  \mathbf{I}_M \right), \\
\label{eq: Estimation-Error-Distribution}
 \mathbf{e}_{l,k} & \sim \mathcal{CN} \left(  \mathbf{0}, \frac{ \beta_{l,k} \sigma_{ \mathrm{UL} }^2 }{ p_k \tau_p \beta_{l,k} + \sigma_{ \mathrm{UL} }^2 }  \mathbf{I}_M  \right).
\end{align}

These distributions provide the statistical properties of the channel estimates that are needed to analyze utility functions like the DL ergodic SE in Massive MIMO cellular systems and further formulate joint user association and QoS optimization problems, which are the main goals of this paper.
\vspace*{-0.1cm}
\subsection{Data Transmission Model}
\vspace*{-0.2cm}
We assume that each BS is allowed to transmit to each user but sends a different data symbol than the other BSs. This is referred to as non-coherent joint transmission \cite{Li2015} and it is less complicated to implement than coherent joint transmission which requires phase-synchronization between BSs. At BS $l$, the transmit signal $\mathbf{x}_l$ is 
\begin{equation}
\mathbf{x}_l = \sum_{ t =1 }^{K} \sqrt{ \rho_{l,t} } \mathbf{w}_{l,t} s_{l,t}.
\end{equation}
Here the scalar data symbol $s_{l,t}$, which BS $l$ intends to transmit to user $t$, has unit power $\mathbb{E} \{ | s_{l,t }|^2 \} = 1$ and $\rho_{l,t}$ stands for the transmit power allocated to this particular user. In addition, the corresponding linear precoding vector $ \mathbf{w}_{l,t} \in \mathbb{C}^{M}$ determines the spatial directivity of the signal sent to this user.  We notice that user $t$ is associated with BS $l$ if and only if $\rho_{l,t} \neq 0$, and each user can be associated with multiple BSs. The received signal at an arbitrary user $k$ is modeled as
\begin{equation} \label{eq: Downlink-Signal}
 y_k = \sum\limits_{i = 1}^{L} \sqrt{ \rho_{i,k} } \mathbf{h}_{i,k}^H \mathbf{w}_{i,k} s_{i,k} + \sum\limits_{i =1 }^{L} \sum\limits_{ \substack{t =1 \\ t \neq k} }^{K} \sqrt{ \rho_{i,t} } \mathbf{h}_{i,k}^H \mathbf{w}_{i,t} s_{i,t} + n_k. 
\end{equation}
The first part in \eqref{eq: Downlink-Signal} is the superposition of desired signals that user $k$ receives from the BSs. The second part is multi-user interference that degrades the quality of the detected signals. The third part is the additive white noise $n_k \sim \mathcal{CN} (0, \sigma_{ \mathrm{DL} }^2 )$. 

To avoid spending precious DL resources on pilot signaling, we suppose that user $k$ does not have any information about the current channel realizations but only knows the channel statistics \cite{Bjornson2016b}. User $k$ would like to detect all the desired signals coming from the BSs. To achieve low computational complexity, we assume that each user detects its different data signals sequentially and applies successive interference cancellation. From these assumptions, a lower bound on the capacity between the BSs and user $k$ is given in Theorem \ref{Theorem-Lower-Bound-Rate}. 

\begin{theorem} \label{Theorem-Lower-Bound-Rate}
By decoding Gaussian signals in a successive manner, a lower bound on the DL ergodic sum capacity of an arbitrary user $k$ is given by
\begin{equation} \label{eq: Sum-Rate-k}
R_k =  \left( 1 - \frac{\tau_p}{\tau_c} \right) \log_2 \left(1 + \mathrm{SINR}_k \right) \quad \textrm{[bit/symbol]},
\end{equation}
where the value of the SINR is defined in \eqref{eq: SINR_k}.
\begin{figure*}[t]
\begin{equation} \label{eq: SINR_k}
\mathrm{SINR}_k  = \frac{ \sum_{ i=1}^{L} \rho_{i,k} | \mathbb{E} \{ \mathbf{h}_{i,k}^H \mathbf{w}_{i,k} \} |^2 }{ \sum_{ i =1}^{L}  \rho_{i,k} (\mathbb{E} \{ | \mathbf{h}_{i,k}^H \mathbf{w}_{i,k} |^2 \} - | \mathbb{E} \{ \mathbf{h}_{i,k}^H \mathbf{w}_{i,k} \} |^2 ) + \sum_{ i=1 }^{L} \sum_{ \substack{t =1 \\ t \neq k} }^{K} \rho_{i,t} \mathbb{E} \{ | \mathbf{h}_{i,k}^H \mathbf{w}_{i,t} |^2 \} + \sigma_{ \mathrm{DL} }^2 }
\end{equation}
\vspace*{-0.1cm}
\hrulefill
\vspace*{-0.4cm}
\end{figure*}
\end{theorem}
\begin{proof}
The sum SE $R_k$ is obtained, similar to \cite{Tse2005a, Li2015}, by letting user $k$ detect the received signals coming from all BSs. Suppose that user $k$ is currently detecting the signal sent by an arbitrary BS $i$, say $s_{i,k}$, and possesses the detected signals of the $(i-1)$ previous BSs  but not their instantaneous channel realizations. Let us  denote $R_{i,k}$ the lower bound on the ergodic SE between user $k$ and BS $i$, and therefore $R_k = \sum_{i=1}^{L} R_{i,k}$. The detailed proof is available in \cite{Chien2015}.
\end{proof}

The numerator in \eqref{eq: SINR_k} is a summation of the desired signal powers sent to user $k$ over the average precoded channels from each BS. It confirms that BS cooperation in the form of non-coherent joint transmission is feasible and has the potential to increase the sum SE at the users since all signals are useful. The first term in the denominator represents beamforming gain uncertainty, caused by the lack of CSI at the terminal, while the second term is multi-user interference and the third term represents Gaussian noise. Besides, we stress that the SE expression in Theorem \ref{Theorem-Lower-Bound-Rate} holds for any channel distribution and precoding schemes. 
\vspace*{-0.2cm}
\subsection{Achievable Spectral Efficiency under Rayleigh Fading} \label{Achievable-Spectral-Efficiency}
\vspace*{-0.1cm}
 We consider MRT precoding which is defined as
\begin{equation}  \label{eq: Linear-Precoding-Vector}
\mathbf{w}_{l,k}  =  \frac{ \hat{\mathbf{h}}_{l,k} }{ \sqrt{ \mathbb{E} \{  \| \hat{ \mathbf{h}}_{l,k}  \|^{2} \}}}.
\end{equation}
The lower bound on the ergodic capacity is obtained in closed form with MRT, as shown in Corollary \ref{Corollary-MRT-Rate}.
\begin{corollary} \label{Corollary-MRT-Rate}
For Rayleigh fading channels, if the BSs utilize MRT precoding, then the lower bound on the DL ergodic sum rate in Theorem \ref{Theorem-Lower-Bound-Rate} is simplified to
\begin{equation} \label{eq: Corollary-MRT-Rate}
R_k = \left( 1 - \frac{\tau_p}{\tau_c} \right) \log_2\left( 1 + \mathrm{SINR}_k \right) \; \textrm{[bit/symbol]},
\end{equation}
where the SINR is
\begin{equation} \label{eq: SINR-MRT}
\mathrm{SINR}_k = \frac{ M \sum_{i =1}^{L} \rho_{i,k} \frac{p_k \tau_p \beta_{i,k}^2 }{ p_k \tau_p \beta_{i,k} + \sigma_{\mathrm{UL}}^2 } }{ \sum_{i = 1 }^{L} \sum_{ t=1 }^{K} \rho_{i,t} \beta_{i,k} + \sigma_{ \mathrm{DL}}^2 }.
\end{equation}
\begin{proof}
We exploit closed-form expressions for the moments of circularly symmetric Gaussian variables in order to compute the expectations in \eqref{eq: SINR_k}. The full proof is available in \cite{Chien2015}.
\end{proof}
\end{corollary}
This corollary reveals that the signal power increases proportionally to $M$ thanks to the array gain. Meanwhile, the interference is unaffected by the number of BS antennas. In addition, the non-coherent combination of received signals at user $k$ adds up the power from multiple BSs and can give stronger signal gain than if only one BS serves the user.

\section{Joint Total Transmit Power Optimization and User Association Optimization} \label{Section:Power-Optimization}
\vspace*{-0.1cm}
\subsection{Problem Formulation}\label{Sec: Power-Consumption}
\vspace*{-0.2cm}
The transmit power at an arbitrary BS $i$, $P_{ \mathrm{trans},i}$ is limited by the peak radio frequency output power $P_{\mathrm{max},i}$, which defines the maximum power that can be utilized at each BS. In particular, $P_{ \mathrm{trans},  i}$ is computed from the transmit signals as
\begin{equation} \label{eq: Transmit-Power}
 P_{\mathrm{trans},i} =  \mathbb{E} \{ \| \mathbf{x}_i \|^2 \} =  \sum_{t=1 }^{K} \rho_{i,t},  \; 0 \leq P_{\textrm{trans},i} \leq P_{\textrm{max},i}.
\end{equation}
The main goal of a Massive MIMO network is to deliver a promised service quality to the users, while consuming as little power as possible. In this paper, we formulate this as a power minimization problem under user-specific QoS constraints as
\begin{equation} \label{Optimization: General-Form}
\begin{aligned}
& \underset{ \{\rho_{i,t} \geq 0 \} }{\textrm{minimize}}
& & \sum_{ i=1 }^{L} P_{\textrm{trans},i} \\
& \textrm{subject to}
& &  R_k \geq \xi_k,\; \forall k \\
& &&  P_{\mathrm{trans},i} \leq P_{\mathrm{max}, i }, \; \forall i,\\
\end{aligned}
\end{equation}
where $\xi_k$ represents for the target QoS of user $k$. By defining $\hat{\xi}_k = 2^{ \frac{\xi_k \tau_c }{(\tau_c - \tau_p)} } -1$ and plugging \eqref{eq: Corollary-MRT-Rate} and \eqref{eq: Transmit-Power} into \eqref{Optimization: General-Form}, the power minimization problem for MRT is expressed as
\begin{equation} \label{eq:PowerMRT}
\begin{aligned}
& \underset{ \{ \rho_{i,t} \geq 0 \}}{\mathrm{minimize}}
& & \sum_{i=1}^{L} \sum_{t=1}^{K} \rho_{i,t} \\
& \mbox{subject to}
& & \frac{ \sum_{i=1}^{L} \rho_{i,k} \frac{ M p_k \tau_p  \beta_{i,k}^2 }{ p_k \tau_p \beta_{i,k} + \sigma_{ \mathrm{UL} }^2 } }{ \sum_{i=1}^{L} \sum_{t=1}^{K}  \rho_{i,t} \beta_{i,k} +\sigma_{ \mathrm{DL} }^2 }  \geq \hat{ \xi }_k, \; \forall k\\
& && \sum_{t=1}^{K} \rho_{i,t} \leq P_{\mathrm{max},i}, \; \forall i. \\
\end{aligned}
\end{equation}

The jointly optimal power allocation and user association are obtained by solving this problem. At the optimal solution, each user $t$ in the network is associated with the subset of BSs that is determined by the non-zero values $\rho_{i,t}$. There are fundamental differences between our problem formulation and previous ones such as \cite{Li2015} and the references therein. The main distinction is that these previous works consider short-term QoS constraints that depend on the current fading realizations, while we consider long-term QoS constraints that do not depend on instantaneous fading realizations thanks to channel hardening and favorable properties in Massive MIMO.

\vspace*{-0.3cm}
\subsection{Optimal Solution with Linear Programming}
\vspace*{-0.2cm}
Let us denote the power control vector of an arbitrary user $t$ by $\pmb{\rho}_t = [\rho_{1, t}, \ldots, \rho_{L,t}]^T \in \mathbb{C}^{L}$, where its entries satisfy $\rho_{i,t} \geq 0$ meaning that $ \pmb{\rho}_t  \succeq 0$. We also denote $\mathbf{1} = [1,\ldots,1]^T \in \mathbb{C}^{L}$ and $\pmb{\epsilon}_i \in \mathbb{C}^L$ has all zero entries but the $i$th one is 1. The optimal power allocation is obtained as shown in Theorem \ref{Theorem: Linear-Solution}.
\begin{theorem} \label{Theorem: Linear-Solution}
The optimal solution to the total transmit power minimization problem in \eqref{eq:PowerMRT} is obtained by solving the following linear program:
\begin{equation} \label{eq: Linear-Solution-CVX}
\begin{aligned}
& \underset{ \{ \pmb{\rho}_t \succeq 0 \}  }{\mathrm{minimize}}
& & \sum_{t=1}^{K} \mathbf{1}^T \pmb{\rho}_t \\
& \textrm{subject to}
& & \sum_{t=1}^{K} \mathbf{c}_k^T \pmb{\rho}_t -  \mathbf{b}_k^T \pmb{\rho}_k + \sigma_{ \mathrm{DL} }^2 \leq 0, \; \forall k \\
& && \sum_{t=1}^{K} \pmb{\epsilon}_i^T \pmb{\rho}_{t} \leq P_{\mathrm{max},i}, \; \forall i.\\
\end{aligned}
\end{equation}
Here, the vectors $\mathbf{c}_{k} $ and $\mathbf{b}_{k} $ are defined as
\begin{equation*}
\begin{split}
\mathbf{c}_k &= \left[ \beta_{1,k}, \ldots, \beta_{L,k}  \right]^T,\\
\mathbf{b}_k &= \left[ \frac{M p_k \tau_p \beta_{1,k}^2 }{ \hat{\xi}_k \left( p_k  \tau_p \beta_{1,k} + \sigma_{ \mathrm{UL} }^2 \right) }, \ldots, \frac{M p_k \tau_p \beta_{L,k}^2 }{\hat{\xi}_k  \left( p_k \tau_p \beta_{L,k} + \sigma_{ \mathrm{UL} }^2  \right) } \right]^T.
\end{split}
\end{equation*}
\end{theorem}
\begin{proof}
The problem in \eqref{eq: Linear-Solution-CVX} is obtained from \eqref{eq:PowerMRT} after some algebra. We note that the objective function is a linear combination of power vectors $\pmb{\rho}_t$, for $t=1, \ldots, K$. Moreover the constraint functions are affine functions of power variables. Thus, the optimization problem \eqref{eq: Linear-Solution-CVX} is a linear program.
\end{proof}

The merits of Theorem \ref{Theorem: Linear-Solution} are twofold: The total transmit power minimization problem for a Massive MIMO cellular system with non-coherent joint transmission can be solved to global optimality in polynomial time. The optimal power allocation can thus be obtained by interior-point methods, for example, general-purpose implementations such as CVX \cite{cvx2015}. Additionally, the solution provides the optimal BS-user association in the system as presented in the next subsection.
\vspace*{-0.2cm}
\subsection{BS-User Association Principle}
\vspace*{-0.2cm}
To shed light on the optimal association between users and BSs provided by the solution in Theorem \ref{Theorem: Linear-Solution}, we use Lagrange duality theory. The Lagrangian of \eqref{eq: Linear-Solution-CVX} is
\begin{equation} \label{eq: Langrangian}
\begin{split}
& \mathcal{L} \left( \pmb{ \rho}_t, \lambda_k, \mu_i \right) = \\
& \sum_{ t= 1}^{K} \mathbf{1}^T \pmb{\rho}_t+ \sum_{k=1}^{K} \lambda_k\left( \sum_{t=1}^{K} \mathbf{c}_k^T \pmb{\rho}_t - \mathbf{b}_k^T \pmb{\rho}_k + \sigma_{\mathrm{DL}}^2\right)\\
& +
 \sum_{i=1}^{L}\mu_i \left( \sum_{t=1}^{K} \pmb{\epsilon}_i^T \pmb{\rho}_t - P_{\mathrm{max},i}\right),
\end{split}
\end{equation}
where the non-negative Lagrange multipliers $\lambda_k$ and $\mu_i$ are associated with the $k$th QoS constraint and the peak transmit power constraint at BS $i$, respectively. Based on the Lagrangian we can formulate the dual problem and then obtain the set of BSs that serves an 
arbitrary user $t$ as follows.
\begin{theorem} \label{Theorem-BS-Association}
Let $\{ \check{\lambda}_k, \check{\mu}_i \}$ denote the optimal Lagrange multipliers. User $t$ is served only by the subset of BSs with indices in the set
\begin{equation} \label{eq: BS-Association}
\mathcal{S}_t = \underset{i}{ \mathrm{argmin}}  \left(1 + \sum\limits_{k=1}^{K} \check{\lambda}_{k} c_{i,k} + \sum\limits_{i=1}^{L} \check{\mu}_{i} \right)\frac{ 1 }{b_{i,t}},
\end{equation}
where $c_{i,k} = \beta_{i,k} $ and $b_{i,t} = \frac{M p_t \tau_p \beta_{i,t}^2 }{ \hat{\xi}_t ( p_t  \tau_p \beta_{i,t} + \sigma_{ \mathrm{UL} }^2 ) }$. 

The optimal BS association for user $t$ is further specified as one of the following two cases:
\begin{itemize}
\item It is served by one BS if the set $\mathcal{S}_t$ in \eqref{eq: BS-Association} only contains one index.
\item It is served by a subset of BSs if the set $\mathcal{S}_t$ in \eqref{eq: BS-Association} contains several indices.
\end{itemize}
\end{theorem}
\begin{proof}
Conditions on the optimal Lagrange multipliers are obtained by taking the first-order derivative of the Lagrange function with respect to power variables. Thereafter \eqref{eq: BS-Association} is obtained by ensuring that the dual problem is bounded from below. The full proof is available in \cite{Chien2015}.
\end{proof}

The expression in \eqref{eq: BS-Association} explicitly shows that the optimal BS-user association is affected by many factors such as interference between BSs, noise variance, power allocation, large-scale fading, channel estimation quality and QoS constraints. There is no simple association rule since the function depends on the Lagrange multipliers, but we can be sure that the max-SNR association is not always optimal.

\section{Max-min QoS Optimization} \label{Section:Max-Min-QoS}

This section is inspired by the fact that there is not always a feasible solution to the power minimization problem with fixed QoS constraints in \eqref{eq: Linear-Solution-CVX}. For a certain network, it is not easy to select the target QoS values. Our vision is to supply a good target QoS for all users by maximizing the lowest QoS value, possibly with some user specific weighting \cite{Yang2014a}. Consequently, we formulate the max-min QoS problem as
\begin{equation} \label{Problem: Max-Min-QoS}
\begin{aligned}
& \underset{ \{ \rho_{i,t} \geq 0 \} }{\textrm{maximize}} \; \underset{k}{\textrm{min}}
& &  R_k / w_k  \\
& \textrm{subject to}
& & P_{ \mathrm{trans},i} \leq P_{\mathrm{max},i} \;, \forall i  ,\\
\end{aligned}
\end{equation}
where $w_k > 0$ is the specific weight for user $k$. The weights can be assigned based on for example information about the propagation, interference situation and user priorities. If there is no prior information regarding the users, they may be set to $1$. To solve \eqref{Problem: Max-Min-QoS}, it is converted to the epigraph form as
\begin{equation} \label{eq: WeightSpecific2}
\begin{aligned}
& \underset{ \{ \rho_{i,t} \geq 0 \}, \xi}{\textrm{maximize}}
& & \xi \\
& \textrm{subject to}
& & R_{k} / w_k  \geq \xi \;, \forall k \\
& & & P_{ \mathrm{trans},i} \leq P_{\mathrm{max},i} \;, \forall i,\\
\end{aligned}
\end{equation}
where $\xi$ is the minimum QoS parameter for the users that we aim to maximize. Note that we can solve \eqref{eq: WeightSpecific2} for a fixed $\xi$ as a linear program using Theorem \ref{Theorem: Linear-Solution} with $\xi_k = \xi w_k$. Since the QoS constraints are increasing functions of $\xi$, the solution to the max-min QoS problem can be obtained by doing a line search over the range $\mathcal{R}=[0,\xi_{0}^{\mathrm{upper}}]$, where $\xi_{0}^{\mathrm{upper}}$ is selected to make \eqref{eq: Linear-Solution-CVX} infeasible, to get the maximal value. Hence, this is a quasi-linear program. Note that we jointly maximize the minimum QoS and find the optimal BS-user association when solving \eqref{Problem: Max-Min-QoS} and \eqref{eq: WeightSpecific2}.

We use the bisection line search method \cite{Bjornson2013d} to obtain the solution. The problem \eqref{eq: WeightSpecific2} is solved in an iterative manner: by iteratively reducing the size of the search range and solve the problem \eqref{eq: Linear-Solution-CVX}, the maximum QoS level and optimal BS-user association can be simultaneously optimized. At each iteration, the feasibility of \eqref{eq: Linear-Solution-CVX} is verified with the  value $\xi^{\mathrm{candidate}} \in \mathcal{R}$, that is defined as the middle point of the current search range. If the problem is feasible, then its solution $\{\check{\pmb{\rho}}_k \}$, for $k=1, \ldots, K$, is assigned as the power allocation and the lower bound $\xi^{\textrm{lower}}$ is updated as well. Otherwise, if the problem is infeasible, then a new upper bound is set up. The algorithm will be terminated if the difference between the upper and lower limit of the search range is smaller than a line-search accuracy value $\delta$. 
The max-min QoS optimization is summarized in Algorithm \ref{Algorithm: Bisection}.
\vspace*{-0.3cm}
\begin{algorithm}[h]
\caption{Max-min QoS based on the bisection method}
\textbf{Result:} Solve optimization in \eqref{Problem: Max-Min-QoS}. 
\\ \textbf{Input:}  Initial upper bound $\xi_0^{\mathrm{upper}}$, and line-search accuracy $\delta$;
\begin{algorithmic}
\State Set  $\xi^{\mathrm{lower}}= 0$; $\xi^{\mathrm{upper}}= \xi_0^{\mathrm{upper}}$;
\While {$\xi^{\mathrm{upper}} - \xi^{\mathrm{lower}} > \delta$}
\State Set  $\xi^{\mathrm{candidate}} = \frac{\xi^{\mathrm{upper}}+\xi^{\mathrm{lower}}}{2}$;
\If { \eqref{eq: Linear-Solution-CVX} is infeasible for $ \xi_k = w_k \xi^{\mathrm{candidate}}, \forall k, $ } \do
\\
 \State Set $\xi^{\mathrm{upper}} = \xi^{\mathrm{candidate}}$;
\Else

\State Set $ \{ \check{\pmb{\rho}}_{k} \}$ as the solution to \eqref{eq: Linear-Solution-CVX};

\State Set $\xi^{\mathrm{lower}} = \xi^{\mathrm{candidate}}$ ;
\EndIf
\EndWhile
\State Set $\xi_{\mathrm{final}}^{\mathrm{lower}} = \xi^{\mathrm{lower}}$ and $\xi_{\mathrm{final}}^{\mathrm{upper}} = \xi^{\mathrm{upper}}$;
\end{algorithmic}
\textbf{Output:} Final interval $[\xi_{\mathrm{final}}^{\mathrm{lower}}, \xi_{\mathrm{final}}^{\mathrm{upper}}]$ and $ \{ \pmb{\rho}_{k} \}= \{ \check{\pmb{\rho}}_{k} \}$;
\label{Algorithm: Bisection}
\end{algorithm} 

\vspace*{-0.4cm}

\section{Numerical Results} \label{Section:Numerical-Results}

In this section, the analytical contributions from the previous sections are evaluated by simulation results for a Massive MIMO cellular system with MRT precoding. Our system deployment is sketched in Fig.~\ref{fig:MassiveMIMOSystem-Layout} comprising of $4$ BSs and $20$ users. For the max-min QoS algorithm, user specific weights are set to $w_k =1$, $ \forall k$, to make it easy to interpret the results. Since the joint power allocation and BS-user association obtains the optimal subset of BSs that associate with users, we denote it to as ``Optimal" in the figures. For comparison, we also consider a greedy method, in which each user is associated with only the BS that gives the strongest signal on the average (i.e., the max-SNR value). The performance is averaged over user locations. 

The peak DL radio frequency output power is $40$ W per BS. The system bandwidth is $20$ MHz and the coherence interval is of $200$ symbols. The users send the orthogonal pilot sequences whose length equals the number of users (i.e., 20 symbols) and has energy of $2\times 10^{-7}$ J. The shadow fading $z_{l,k}$ is generated by utilizing log-normal Gaussian distribution with standard deviation $7$ dB. The path loss at distance $d$ km is calculated by $148.1 + 37.6 \log_{10}d$. Consequently, the channel variance $\beta_{l,k}$ is computed as $\beta_{l,k} = -148.1 - 37.6 \log_{10} d + z_{l,k} $ dB. Moreover, with the noise figure of $5$ dB, the noise variance for both the UL and DL is $-96$ dBm.
\begin{figure}[t]
    \centering
    \includegraphics[trim=12.3cm 8.2cm 0.9cm 9.7cm, clip=true, width=1.9 in]{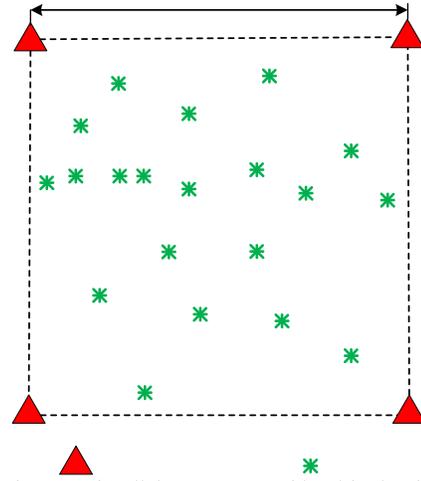}
    \caption{ Massive MIMO cellular system considered in the simulations: The BS locations are fixed, while $20$ users are randomly distributed over the joint coverage area of the BSs.}
    \label{fig:MassiveMIMOSystem-Layout}
    \vspace{-0.20cm}
\end{figure} 
\begin{figure}[t]
     \centering
     \includegraphics[trim=2.5cm 0.8cm 2cm 1cm, clip=true, width=2.76 in]{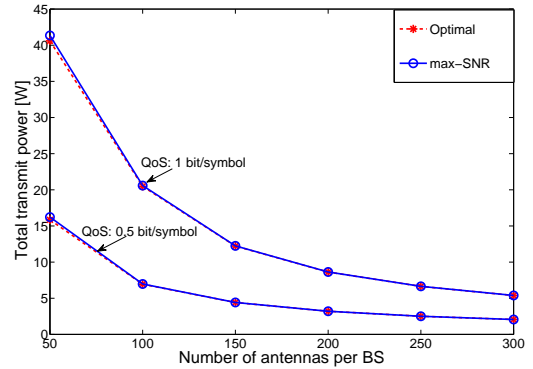}
     \caption{The total transmit power $(\sum_{i=1}^L P_i)$ versus the number of antennas at each BS.}
     \label{Fig-PowervsAntenna}
     \vspace{-0.4cm}
 \end{figure}
 \begin{figure}[t]
      \centering
      \includegraphics[trim=2.1cm 0.4cm 2cm 1.2cm, clip=true, width=2.76in]{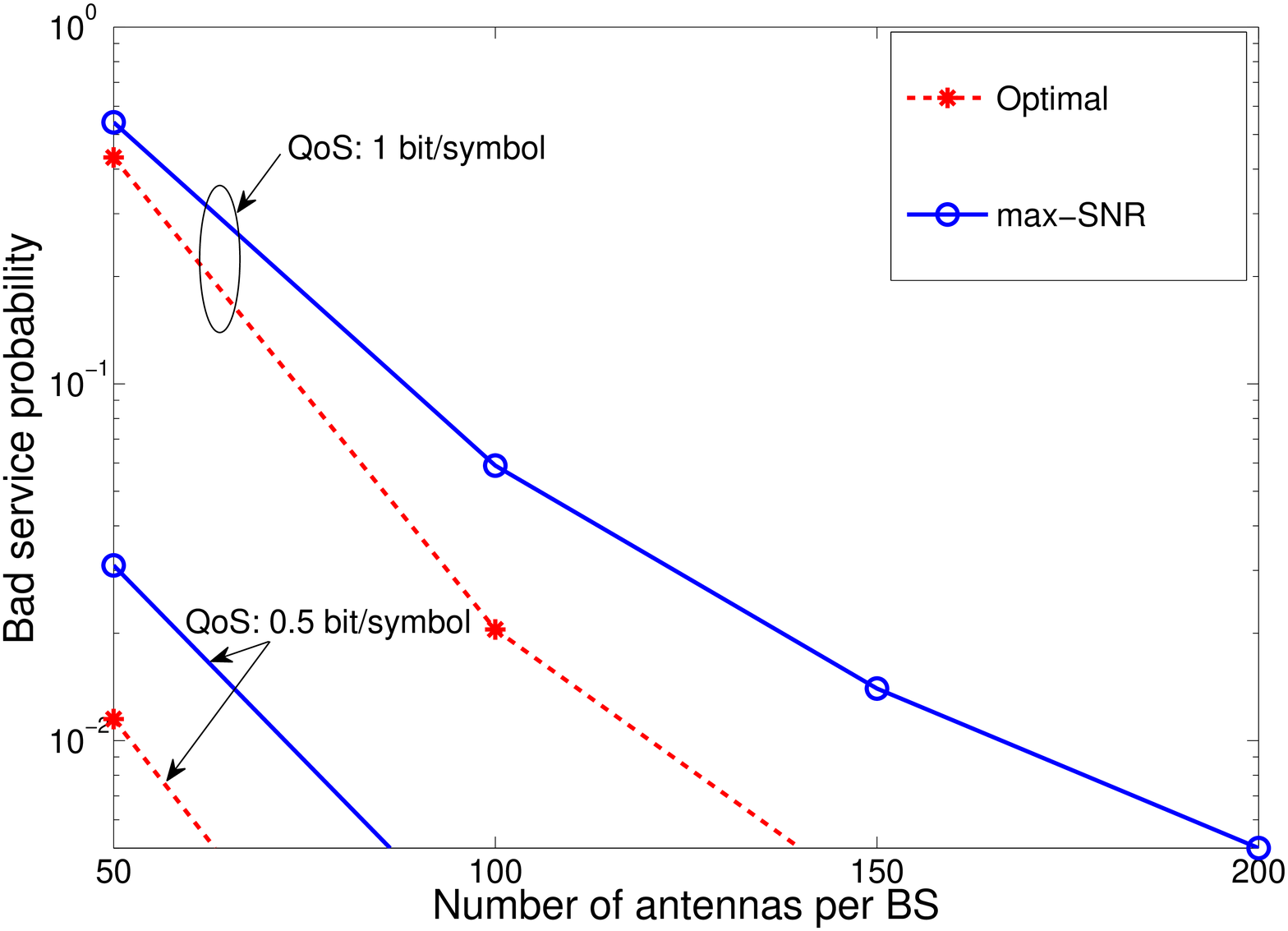}
      \caption{The bad service probability versus the number of BS antennas.}
      \label{Fig-InfeasibleAntennas}
      \vspace{-0.4cm}
  \end{figure}
\vspace*{-0.1cm}

We show the total transmit power as a function of the number of BS antennas in Fig.~\ref{Fig-PowervsAntenna}. For fair comparison, the results are averaged over only user locations where both association methods can satisfy the QoS constraints. The results reveal a superior reduction of the total transmit power compared to the peak one, say $160$ W, in small-scale MIMO cellular systems. In addition, we notice that the simple max-SNR association is close to optimal in these cases.

The difference between the optimal association and max-SNR association is seen from the fact that sometimes only the former can satisfy the QoS constraints. Fig.~\ref{Fig-InfeasibleAntennas} demonstrates the ``bad service probability" which is defined as the fraction of random user locations and shadow fading realizations in which the power minimization problem in Theorem \ref{Theorem: Linear-Solution} is infeasible. The optimal BS-user association is more robust to environment variations than the max-SNR association since the non-coherent joint transmission can help to resolve the infeasibility. In addition, the figure also verifies difficulties of providing high target SE for all the users, especially when the BSs have a small number of antennas or the users demand high QoS levels. This is a key motivation to consider the max-min QoS optimization problem instead, because it provides feasible solutions for any user locations and channel realizations.

  \begin{figure}[t]
        \centering
        \includegraphics[trim=2.0cm 0.8cm 2cm 1cm, clip=true, width=2.76 in]{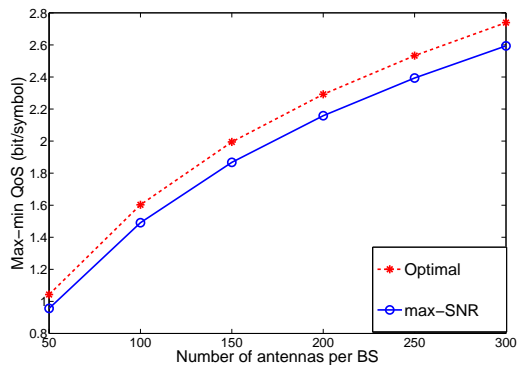}
        \caption{The max-min QoS level versus the number of BS antennas.}
        \label{Fig-PowerMaxMinQoSMRT}
        \vspace{-0.4cm}
    \end{figure}
    
  \begin{figure}[t]
         \centering
         \includegraphics[trim=2.0cm 0.8cm 2cm 1cm, clip=true, width=2.76 in]{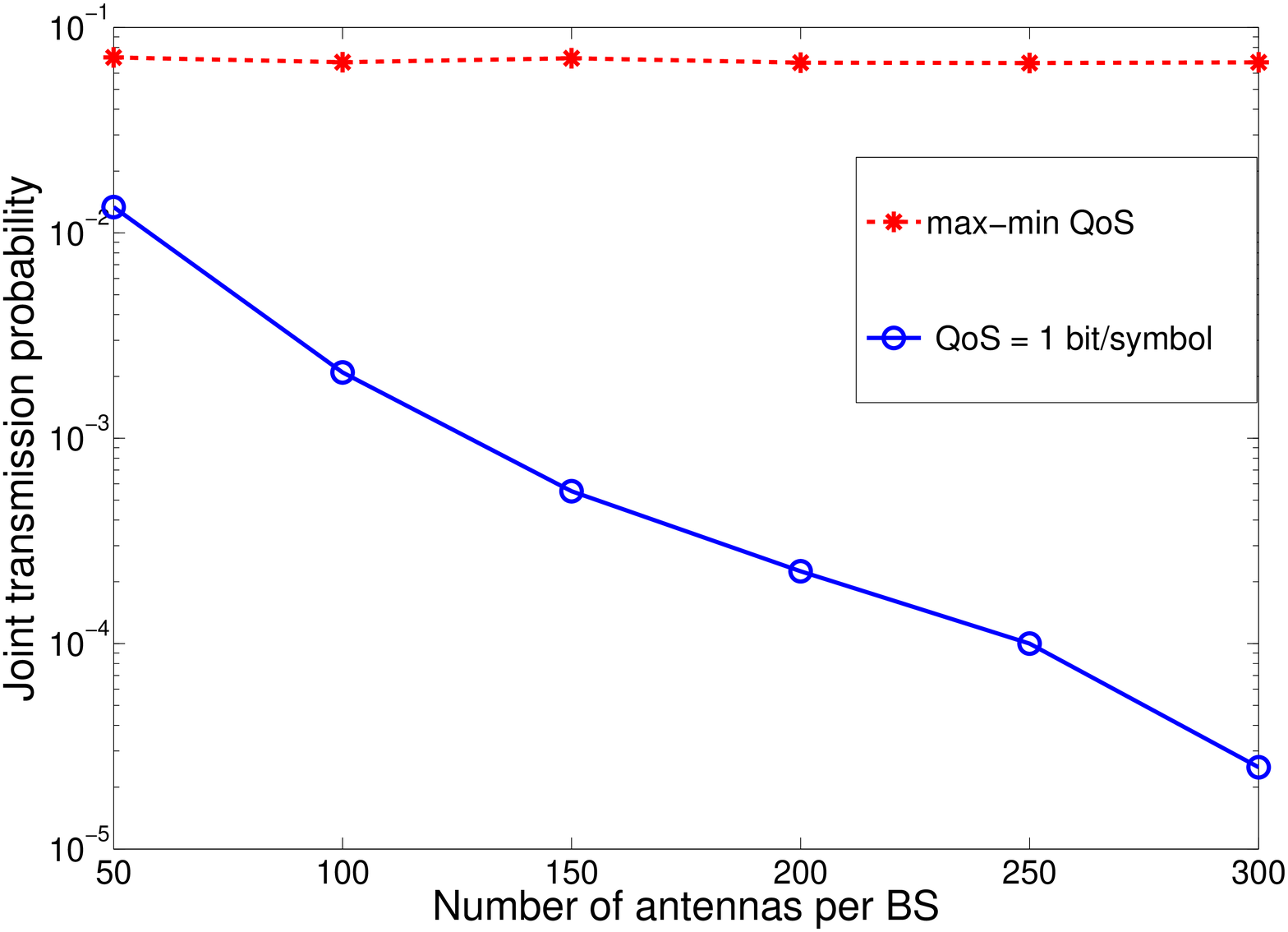}
         \caption{ The joint transmission probability versus the number of BS antennas.}
         \label{Fig-BSAssociationProbability}
         \vspace{-0.4cm}
     \end{figure}
     
Fig.~\ref{Fig-PowerMaxMinQoSMRT} shows the achievable max-min QoS level for all users as a function of the number of BS antennas. Massive MIMO is able to provide good QoS to every user. Roughly speaking, the optimal user association provides up to $10\%$ higher QoS than the max-SNR association. Besides, the probability that a user is served by more than one BS is shown in Fig.~\ref{Fig-BSAssociationProbability}. Even though the system model lets BSs cooperate with each other to serve the users, experimental results verify that single-BS association is optimal in $93 \%$ or more of the cases. This result for the Massive MIMO systems is similar to those obtained by the multi-tier heterogeneous network with multiple-antennas at BSs in \cite{Li2015}. Although the joint transmission probability is relatively independent of the number of BS antennas, the max-min SE is significantly improved if the BSs are equipped with massive antennas as shown in Fig.~\ref{Fig-PowerMaxMinQoSMRT}. Case $2$ in Theorem \ref{Theorem-BS-Association} might happen for example, when some users experience severe shadow fading realizations or there is a high user load at some BSs, which can only be resolved by joint transmission from multiple BSs. Thus, only the optimal BS-user association satisfies QoS constraints, while the max-SNR association falls in bad services as shown in Fig.~\ref{Fig-InfeasibleAntennas}.

\section{Conclusion}
The joint power allocation and BS-user association for the DL non-coherent joint transmission in Massive MIMO systems was designed to minimize the total transmit power consumption at the BSs while satisfying QoS constraints at the users. For Rayleigh fading channels, we proved that the total transmit power minimization problem with MRT precoding is a linear program, so it is always solvable to global optimality in polynomial time. Additionally, we provided the optimal BS-user association rule. In order to ensure that all users are fairly treated, we also solved the max-min QoS optimization problem. Numerical results manifested the effectiveness of our proposed methods, and that the max-SNR association works well in many Massive MIMO scenarios but is not optimal.

\vspace*{-0.15 cm}
\bibliographystyle{IEEEtran}
\bibliography{IEEEabrv,refs}
\end{document}